\documentclass[12pt]{amsart}
\usepackage{amsmath,amsfonts,amsbsy,amsgen,amscd,mathrsfs,amssymb,amsthm}
\usepackage{enumerate,mathtools}

\usepackage{clipboard}
\newclipboard{myclipboard}

\usepackage{fullpage}
\usepackage{url}

\usepackage{mathrsfs}

\usepackage[colorlinks=true,linkcolor=blue]{hyperref} 
\makeatletter
\renewcommand*{\eqref}[1]{%
  \hyperref[{#1}]{\textup{\tagform@{\ref*{#1}}}}%
}
\makeatother

\numberwithin{equation}{section}

\newtheorem{theorem}{Theorem}[section]
\newtheorem*{theorem*}{Main theorem}
\newtheorem*{corollary*}{Corollary}

\newtheorem{lemma}[theorem]{Lemma}
\newtheorem{example}[theorem]{Example}
\theoremstyle{definition}
\newtheorem{definition}[theorem]{Definition}

\newtheorem{remark}[theorem]{Remark}

\newtheorem*{claim*}{Claim}

\newcommand{\R}{\mathbb R}
\newcommand{\PP}{\mathbf P}

\newcommand{\EE}{\mathbf E}

\newcommand{\Ent}{\operatorname{Ent}}
\newcommand{\Var}{\operatorname{Var}}
\newcommand{\Id}{\operatorname{Id}}

\newcommand{\T}{\mathsf{T}}

\title{Exact renormalization groups and transportation of measures}

\author{Yair Shenfeld}
\address{Department of Mathematics, Massachusetts Institute of Technology, 
Cambridge, MA, USA}
\email{shenfeld@mit.edu}

\begin{document}
\maketitle

\begin{abstract}
This note provides a new perspective on Polchinski's exact renormalization group, by explaining how it gives rise, via the multiscale Bakry-\'Emery criterion, to Lipschitz transport maps between Gaussian free fields and interacting quantum and statistical field theories. Consequently, many functional inequalities can be verified for the latter field theories, going beyond the current known results.
\end{abstract}

\section{Introduction}
\label{sec:intro}
\subsection{Summary} One of the fundamental tools in the study of quantum and statistical field theories is the renormalization group in its various formulations. We are concerned with an exact version of the renormalization group, due to Polchinski,  where a continuum of scales is used in the renormalization process \cite{polchinski1984renormalization,brydges1987mayer}. Recently, Bauerschmidt \& Bodineau \cite{bauerschmidt2020spectral, BB} and  Bauerschmidt \& Dagallier  \cite{bauerschmidt2022log, bauerschmidt2022logPhi} have shown that the renormalization (semi)group of Polchinski is a valuable tool in proving functional inequalities (Poincar\'e and log-Sobolev) for a number of Euclidean quantum field theories and statistical field theories---the interest in these inequalities stems from the fact that they imply fast relaxation to equilibrium of the dynamics of the respective field theories.  In particular, the application of Polchinski's equation to the study of functional inequalities is  facilitated by the so-called multiscale Bakry-\'Emery criteria. In this work, we  provide a new perspective on the subject by showing how one of the versions of the  multiscale Bakry-\'Emery criteria gives rise to Lipschitz transport maps between Gaussian free fields and  interacting Euclidean quantum, or statistical, (scalar) field theories. The Lipschitz properties of these transport maps imply that many functional inequalities, which are known to hold for Gaussian free fields, also hold (with constants depending on the Lipschitz constants of the transport maps) for any field theory where the multiscale Bakry-\'Emery criterion can be verified. For example, for the two-dimensional massive continuum sine-Gordon model, we recover a log-Sobolev inequality (see \cite{BB}) and prove many other functional inequalities which, until now, were not known (e.g., Theorem \ref{thm:ppoincare}, Theorem \ref{thm:iso}, and Theorem \ref{thm:eigenvalues}). Indeed, the advantage of the transportation of measure approach is that, once a Lipschitz transport map between a free field to an interacting field is constructed, the transfer of functional inequalities between the fields becomes (almost) automatic, and bypasses the need to provide a new proof for each functional inequality of interest. 

\subsection*{Organization} Section \ref{sec:intro} introduces the objects of study in this note and describes the main results. Section \ref{sec:LN} contains the construction and Lipschitz properties of the Langevin/Ornstein-Uhlenbeck transport map, on which our results are based. Section \ref{sec:BM} briefly sketches the connection  between exact renormalization and a different transport map, the Brownian transport map, as well as connections to related works in the literature.

\subsection{Models}
Euclidean quantum, or statistical, field theories can be modeled as formal probability measures on function spaces. Let us present the regularizations of these models. Let $d$ be the dimension and let $L\mathbb T^d$ be the torus in $\R^d$ of side length $L>0$. We let $\Lambda_{\epsilon,L}:=L\mathbb T^d\cap\epsilon \mathbb Z^d$ with $L$ being a multiple of $\epsilon$; here $L$ is the infrared cutoff while $\epsilon$ is the ultraviolet cutoff. Our models will be defined as probability measures  $\nu^{\epsilon,L}$ on $\R^{\Lambda_{\epsilon,L}}$. One of the challenges in constructing such models is to show that the infrared and ultraviolet limits of $\nu^{\epsilon,L}$, that is, $L\uparrow \infty$ and $\epsilon\downarrow 0$, are well-defined. In this work, we will not deal with this issue and, instead, strive for estimates that are independent of $L$ and $\epsilon$. Let us now specify the types of models we focus on. Given a mass $m> 0$ we let $\gamma^{\epsilon,L}$ be the Gaussian free field with covariance $A_{\epsilon}^{-1}:=\epsilon^{-d}(-\Delta^{\epsilon}+m)^{-1}$, that is,
\[
\gamma^{\epsilon,L}(d\varphi)\propto \exp\left[-\frac{\epsilon^d}{2}\sum_{x\in \Lambda_{\epsilon,L}}\left(\varphi_x(-\Delta^{\epsilon}+m)\varphi_x\right)\right]d\varphi=\exp\left[-\frac{1}{2}(\varphi,A_{\epsilon}\varphi)\right]d\varphi,
\]
where $d\varphi$ is the Lebesgue measure on $\R^{\Lambda_{\epsilon,L}}$, $\Delta^{\epsilon}$ is the discrete Laplacian, i.e., $(\Delta^{\epsilon}\varphi)_x=\epsilon^{-2}\sum_{y\sim x}(\varphi_y-\varphi_x)$, and $( \cdot,\cdot)$ is the standard inner product with $|\cdot|:=\sqrt{( \cdot,\cdot)}$. The models $\nu^{\epsilon,L}$ take the form
\[
\nu^{\epsilon,L}(d\varphi)= e^{-V_0^{\epsilon,L}(\varphi)}\gamma^{\epsilon,L}(d\varphi),
\]
for $V_0^{\epsilon,L}:\R^{\Lambda_{\epsilon,L}}\to \R$ (the meaning of the subscript on $V_0^{\epsilon,L}$ will become clear). For example, the \textbf{$2$-dimensional massive sine-Gordon model} is of the form,
\[
V_0^{\epsilon,L}(\varphi)\propto -\sum_{x\in \Lambda_{\epsilon,L}}2z\epsilon^{2-\beta/4\pi}\cos (\sqrt{\beta}\varphi_x),
\]
where $z$ is the coupling constant and $\beta$ is the inverse temperature \cite{BB}.

\subsection{Transportation of measure and functional inequalities} 
\label{subsec:transportfunct}

The idea at the base of our work is that the existence of a Lipschitz transport map between some nice measure $\mu^{\epsilon,L}$ and $\nu^{\epsilon,L}$ provides a systematic way to transfer functional inequalities from $\mu^{\epsilon,L}$ to $\nu^{\epsilon,L}$ \cite{cordero2002some}. Let us demonstrate, formally, this idea with the Poincar\'e inequality, but emphasize that the strength of the transport method is that it applies to many other functional inequalities. Suppose that $\mu^{\epsilon,L}$ is a measure on some measurable space $(\Omega^{\epsilon,L}, \mathcal F^{\epsilon,L})$ which satisfies the Poincar\'e inequality with constant $a^{\epsilon,L}$:  For any nice-enough function $F:\Omega^{\epsilon,L}\to \R$, 
\[
\Var_{\mu^{\epsilon,L}}[F]\le a^{\epsilon,L} \int_{\Omega^{\epsilon,L}} |DF|^2d\mu^{\epsilon,L},
\]
where $DF$ is an appropriate notion of derivative and $|\cdot|$ is an appropriate norm. Let $T^{{\epsilon,L}}:\Omega^{\epsilon,L}\to \R^{\Lambda_{\epsilon,L}}$ be a map which pushes forward $\mu^{\epsilon,L}$ to $\nu^{\epsilon,L}$, and is $c^{\epsilon,L}$-Lipschitz in the sense that $|DT^{{\epsilon,L}}|\le c^{\epsilon,L}$. Then, by the chain rule, for any nice-enough function $F:\R^{\Lambda_{\epsilon,L}}\to \R$,
\begin{align*}
\Var_{\nu^{\epsilon,L}}[F]&=\Var_{\mu^{\epsilon,L}}[F\circ T^{{\epsilon,L}}]\le a^{\epsilon,L}\int_{\Omega^{\epsilon,L}} |D (F\circ T^{\epsilon,L})|^2d\mu^{\epsilon,L}\\
&\le a^{\epsilon,L} (c^{\epsilon,L})^2 \int_{\Omega^{\epsilon,L}} |(\nabla F)\circ T^{\epsilon,L}|^2d\mu^{\epsilon,L}=a^{\epsilon,L}(c^{\epsilon,L})^2\int_{\R^{\Lambda_{\epsilon,L}}} |\nabla F|^2d\nu^{\epsilon,L}.
\end{align*}
We conclude that $\nu^{\epsilon,L}$ satisfies a Poincar\'e inequality with constant $a^{\epsilon,L}(c^{\epsilon,L})^2$. With the above computation in hand, it is clear that, in order to prove functional inequalities for $\nu^{\epsilon,L}$, we should find a measure $\mu^{\epsilon,L}$ and a $c^{\epsilon,L}$-Lipschitz transport map $T^{{\epsilon,L}}$ from $\mu^{\epsilon,L}$ to $\nu^{\epsilon,L}$, such that the constant $a^{\epsilon,L}(c^{\epsilon,L})^2$ is well-behaved with respect to $\epsilon$ and $L$. This is exactly what we accomplish in this work, by taking advantage of the known results about the multiscale Bakry-\'Emery criteria proven by Bauerschmidt \& Bodineau \cite{bauerschmidt2020spectral, BB} and  Bauerschmidt \& Dagallier  \cite{bauerschmidt2022log, bauerschmidt2022logPhi}. For example, following \cite{MS1,MS2}, one could deduce, for the two dimensional massive sine-Gordon model, $\Psi$-log-Sobolev inequalities (which generalize the Poincar\'e and log-Sobolev inequalities), $p$-Poincar\'e inequalities, isoperimetric comparisons (between $\gamma^{\epsilon,L}$ and $\nu^{\epsilon,L}$), and eigenvalues comparisons (between the generators associated with $\gamma^{\epsilon,L}$ and $\nu^{\epsilon,L}$), as well as other functional inequalities, with constants which are independent of the ultraviolet cutoff and, in certain regimes, of the infrared cutoff---see Example \ref{ex:SG}.

\subsection{The Langevin transport map}
Consider a measure $\mu^{\epsilon,L}$ on $\R^{\Lambda_{\epsilon,L}}$ with it associated Langevin dynamics
\begin{align}
\label{eq:Langevin}
d\Phi_t=\nabla\log \left(\frac{d\mu^{\epsilon,L}}{d\varphi}\right)(\Phi_t)dt+\sqrt{2} dB_t,\quad \Phi_0\sim \nu^{\epsilon,L},
\end{align}
where $(B_t)_{t\ge 0}$ is a standard Brownian motion in $\R^{\Lambda_{\epsilon,L}}$. Let $p_t:=\text{Law}(\Phi_t)$ to get a flow of probability measures interpolating between $p_0=\nu^{\epsilon,L}$ and $p_{\infty}=\mu^{\epsilon,L}$. The flow $(p_t)_{t\ge 0}$ satisfies a continuity (Fokker-Planck) equation,
\[
\partial_t p_t(\varphi)=\nabla\cdot\left(p_t(\varphi)\nabla u_t(\varphi)\right)\quad \forall \varphi\in\R^{\Lambda_{\epsilon,L}}, \quad t\ge 0,
\]
where, for any $t\ge 0$, $\nabla u_t$ is the vector field driving the flow, obtained from some function $u_t:\R^{\Lambda_{\epsilon,L}}\to \R$. Going from the Eulerian to the Lagrangian perspective, the vector field $\nabla u_t$ induces a diffeomorphism $S_t:\R^{\Lambda_{\epsilon,L}}\to \R^{\Lambda_{\epsilon,L}}$, which transports $\nu^{\epsilon,L}=p_0$ into $p_t$,  via the equation 
\[
\partial_t S_t(\varphi)=-\nabla u_t\left(S_t(\varphi)\right),\quad S_0=\Id,
\]
see \cite[Theorem 5.34]{villani2003topics}. Setting $T_t$ to be the inverse of $S_t$, and defining
\[
T^{\epsilon,L}:=\lim_{t\uparrow \infty}T_t,
\]
we get the \emph{Langevin transport map} between $p_{\infty}=\mu^{\epsilon,L}$ to $p_0=\nu^{\epsilon,L}$.

In our setting, we take $\mu^{\epsilon,L}=\gamma^{\epsilon,L}$, the 	Gaussian free field. On a conceptual level, this choice is motivated by viewing interacting field theories as non-linear transformations of Gaussian free fields. On a technical level, the choice of $\mu^{\epsilon,L}=\gamma^{\epsilon,L}$ is beneficial for two reasons:
\begin{itemize}
\item Due to the special Gaussian structure, numerous functional inequalities are known to hold  for $\gamma^{\epsilon,L}$, with a constant $a^{\epsilon,L}$ which is independent of both $\epsilon$ and $L$. This is a manifestation of the dimension-free nature of the Gaussian.

\item The vector field $\nabla u_t$  can be explicitly computed as $\nabla u_t(\varphi)=\nabla \log U_t\left(\frac{d\nu^{\epsilon,L}}{d\gamma^{\epsilon,L}}\right)(\varphi)$, where $(U_t)_{t\ge 0}$ is the Ornstein-Uhlenbeck semigroup associated to the Ornstein-Uhlenbeck dynamics \eqref{eq:Langevin}. When the multiscale Bakry-\'Emery criterion can be verified, the  vector field $\nabla u_t$ can be controlled and, hence, the transport map  $T^{\epsilon,L}$ can be shown to be Lipschitz.
\end{itemize}

\begin{remark}
The multiscale Bakry-\'Emery criteria are used in  \cite{bauerschmidt2020spectral, BB, bauerschmidt2022log, bauerschmidt2022logPhi} to control the Polchinski semigroup, and the proofs of the Poincar\'e and log-Sobolev inequalities for $\nu^{\epsilon,L}$ proceed in the same vein as the general Bakry-\'Emery theory \cite{bakry2014analysis}. The Polchinski semigroup interpolates between a Dirac mass $\delta_0$ and $\nu^{\epsilon,L}$, corresponding to the continuum of scales used in the renormalization procedure. However, this interpolation does not correspond to a transport map on  $\R^{\Lambda_{\epsilon,L}}$ (but see section \ref{sec:BM}), since no transport map can exist between a Dirac mass and a non-trivial model. In contrast, we work with the Ornstein-Uhlenbeck semigroup where we blow up the Dirac mass by scaling it in such a way that it becomes a Gaussian $\gamma^{\epsilon,L}$, and hence can be transported into $\nu^{\epsilon,L}$. Indeed, as mentioned above, the perspective we take in this work is that one way of analyzing properties of various field theories is by viewing them as non-linear transformations of Gaussian free fields. The renormalization flow induces a flow of non-linear transformations $(T_t)$, and if the non-linear transformations (i.e., the transport maps) are Lipschitz, then the field theories along the flow, and in particular in the limit, do not differ by much from the free fields theories. In principle, this perspective can conceivably be implemented without the regularization of the infrared and ultraviolet cutoffs since the Ornstein-Uhlenbeck semigroup is well-defined in infinite dimensions \cite{bogachev1998gaussian}; but the situation in infinite dimensions is more delicate and interesting \cite{faris2001ornstein}. 
\end{remark}

\subsection*{Related literature}
In the context of functional inequalities, the idea of using the Langevin flow to construct transport maps goes back to at least Otto \& Villani \cite{otto2000generalization}. The first to show that the Langevin transport map enjoys Lipschitz properties were Kim \& Milman \cite{kim2012generalization}, and the work of Mikulincer and the author \cite{MS2} substantially extended  the Lipschitz properties of this transport map. Both works used the multiscale Bakry-\'Emery criterion, implicitly. Further Lipschitz properties of the Langevin transport map can be found in the work of Klartag \& Putterman \cite{klartag2021spectral} and Neeman \cite{neeman2022lipschitz}.  While Cotler and Rezchikov \cite{cotler2022renormalization} recently made a connection between optimal transport and exact renormalization groups, it was shown by Tanana \cite{tanana2021comparison} that, in general, the Langevin transport map is not the same as the optimal transport map. Finally, a connection between renormalization and Ornstein-Uhlenbeck semigroups is discussed by Faris in \cite{faris2001ornstein}.

\subsection{The Polchinski equation and multiscale Bakry-\'Emery criteria}
\label{subsec:BE}

Let $\gamma^{\epsilon,L}$ be the Gaussian measure on $\R^{\Lambda_{\epsilon,L}}$ with covariance matrix $A_{\epsilon}^{-1}$. Let $Q_t:=e^{-t\frac{A_{\epsilon}}{2}}$, $\dot{C}_t:=Q_t^2=e^{-tA_{\epsilon}}$, and $C_t:=\int_0^t\dot{C}_sds$ so $C_{\infty}:=\int_0^{\infty}\dot{C}_sds=A_{\epsilon}^{-1}$. With $V_0:=V_0^{\epsilon,L}$, let 
\begin{align}
\label{eq:Vt}
V_t(\varphi):=-\log \EE_{C_t}[e^{-V_0(\varphi+\zeta)}],
\end{align}
where the expectation $\EE_{C_t}$ stands for an expectation with respect to a centered Gaussian measure on $\R^{\Lambda_{\epsilon,L}}$ with covariance $C_t$, taken over the variable $\zeta$. The function $V_t$ satisfies the Polchinski equation \cite[eq. (1.10)]{BB},
\begin{align}
\label{eq:polchinski}
\partial_tV_t=\frac{1}{2}\Delta_{\dot{C}_t}V_t-\frac{1}{2}(\nabla V_t)_{\dot{C}_t}^2,
\end{align}
with the notation, for a matrix $M$ and a function $F:\R^{\Lambda_{\epsilon,L}}\to \R$,
\[
(\varphi,\phi)_M:=\sum_{ij}M_{ij}\varphi_i\phi_j,\quad (\varphi)_M^2:=(\varphi,\varphi)_M,\quad \Delta_MF:=(\nabla,\nabla)_MF;
\]
the subscript $M$ is omitted when $M=\Id$. 

\begin{definition}
\label{def:BE}
A model $\nu^{\epsilon,L}$ is said to satisfy the \emph{multiscale Bakry-\'Emery criterion} if there exist real numbers $\dot{\lambda}_t$ (possibly negative) such that, for any $t\ge 0$ and $\varphi\in \R^{\Lambda_{\epsilon,L}}$,
\[
Q_t\nabla^2V_t(\varphi)Q_t	\succeq \dot{\lambda}_t\Id, \quad\textnormal{where } Q_t=e^{-t\frac{A_{\epsilon}}{2}}. 
\]
\end{definition}

\subsection{Results}
Our main result is that if a model satisfies  the multiscale Bakry-\'Emery criterion of Definition \ref{def:BE}, then the Langevin transport map is Lipschitz. For simplicity, we assume that $\nu^{\epsilon,L}$ is smooth, but this assumption can often be removed by approximation, see \cite{MS2}. 
\begin{theorem*}
Suppose that a smooth model $\nu^{\epsilon,L}$ satisfies the multiscale Bakry-\'Emery criterion of Definition \ref{def:BE}. Then, the Langevin transport map $T^{\epsilon,L}$, which pushes forward $\gamma^{\epsilon,L}$ to $\nu^{\epsilon,L}$, is $\exp\left(-\frac{1}{2}\int_0^{\infty}\dot{\lambda}_tdt\right)$-Lipschitz. 
\end{theorem*}

Once a Lipschitz transport map between $\gamma^{\epsilon,L}$ to $\nu^{\epsilon,L}$ is constructed, there are standard techniques of transferring  functional inequalities from the Gaussian $\gamma^{\epsilon,L}$ to the model $\nu^{\epsilon,L}$. This was demonstrated in section \ref{subsec:transportfunct} for the Poincar\'e inequality, but many other functional inequalities can be transferred. Let us mention a few of them (see \cite{MS2} for some more). In the rest of this section, $\nu^{\epsilon,L}$ is a measure satisfying the multiscale Bakry-\'Emery criterion, $c^{\epsilon,L}:=\exp\left(-\frac{1}{2}\int_0^{\infty}\dot{\lambda}_tds\right)$, and $|A_{\epsilon}^{-1}|_{\text{op}}$ is the maximal eigenvalue of $A_{\epsilon}^{-1}$. The constant $ c^{\epsilon,L}|A_{\epsilon}^{-1}|_{\text{op}}$ is what appears in the various functional inequalities in this section.

\begin{example}{\textnormal{(\textbf{The two-dimensional massive sine-Gordon model})}}
\label{ex:SG}
$~$

Since $d=2$ we have
\[
\gamma^{\epsilon,L}(d\varphi)\propto \exp\left[-\frac{1}{2}\sum_{x\in \Lambda_{\epsilon,L}}\left(\varphi_x(-\Delta^1+\epsilon^2 m)\varphi_x\right)\right]d\varphi=\exp\left[-\frac{1}{2}(\varphi,A_{\epsilon}\varphi)\right]d\varphi,
\]
where $(\Delta^1\varphi)_x=\sum_{y\sim x}(\varphi_y-\varphi_x)$, and where  we note that 
\[
|A_{\epsilon}^{-1}|_{\textnormal{op}}\le \frac{1}{m\epsilon^2}.
\]
With
\[
V_0^{\epsilon,L}(\varphi)\propto -\sum_{x\in \Lambda_{\epsilon,L}}2z\epsilon^{2-\beta/4\pi}\cos (\sqrt{\beta}\varphi_x),
\]
the two-dimensional massive sine-Gordon model reads
\[
\nu^{\epsilon,L}(d\varphi)= e^{-V_0^{\epsilon,L}(\varphi)}\gamma^{\epsilon,L}(d\varphi).
\]
By \cite[Proposition 3.1]{BB}, if $\beta<6\pi$, then the multiscale Bakry-\'Emery criterion of Definition \ref{def:BE} holds with 
\begin{align*}
|\lambda_t|\le \lambda^*,
\end{align*}
where $\lambda_t:=\int_0^t\dot{\lambda}_sds$ and $\lambda^*=\lambda^*(\beta,z,m,L)$ is independent of $\epsilon$.  Moreover, there exists $\delta_{\beta}>0$ such that, if
\begin{align*}
Lm\ge 1\quad\text{and}\quad |z|m^{-2+\beta/4\pi}\le \delta_{\beta},
\end{align*}
then $\lambda^*=O_{\beta}(|z|m^{-2+\beta/4\pi})$ uniformly in $L$. In other words, we get that $c^{\epsilon,L}:=e^{-\frac{\lambda^*}{2}}$ is independent of $\epsilon$ and, in certain regimes, independent of $L$.  Hence,
\begin{align*}
[c^{\epsilon,L}]^2|A_{\epsilon}^{-1}|_{\textnormal{op}}\le \frac{e^{-\lambda^*}}{m\epsilon^2},
\end{align*}
and the constant $\frac{e^{-\lambda^*}}{m\epsilon^2}$ will be the order of the constant appearing in the various functional inequalities in this section. Since the continuum normalization of the Dirichlet form is of order $\epsilon^{-2}$ \cite[equation (1.17) and proof of Theorem 1.6]{BB}, we get that the constant appearing in the functional inequalities is of order $\frac{e^{-\lambda^*}}{m}$, i.e., independent of $\epsilon$, and in certain regimes, independent of $L$. For more information on the convergence of the measures $\nu^{\epsilon,L}$ as $\epsilon\downarrow 0$ and $L\uparrow \infty$, we refer to \cite{BB} and references therein.
\end{example}

We start by establishing $\Psi$-Sobolev inequalities, which are generalizations of the Poincar\'e and log-Sobolev inequalities. 
\begin{definition}
\label{def:PHIsob}
Let $\mathcal I$ be a closed interval (possibly unbounded) and let $\Psi:\mathcal I\to \R$ be a twice-differentiable function. We say that $\Psi$ is a \emph{divergence} if each of the functions $\Psi, \Psi'', -\frac{1}{\Psi''}$ is convex. Given a probability measure $\eta$ on $\R^{\Lambda_{\epsilon,L}}$, and a function $F:\R^{\Lambda_{\epsilon,L}}\to\mathcal I$ satisfying $\int F d\eta\in\mathcal I$, we define
\[
\Ent_{\eta}^{\Psi}(F):=\int_{\R^{\Lambda_{\epsilon,L}}}\Psi(F)d\eta-\Psi\left(\int_{\R^{\Lambda_{\epsilon,L}}}F\,d\eta\right).
\]
\end{definition}
Some classical examples of divergences are $\Psi(x)=x^2$ with $\mathcal I=\R$ (Poincar\'e inequality), $\Psi(x)=x\log x$ with $\mathcal I=\R_{\ge 0}$ (log-Sobolev inequality), and $\Psi(x)=x^p$ with $\mathcal I=\R_{\ge 0}$ and $1<p<2$. 

By \cite[Corollary 9]{chafai2002entropies}, the measure  $\gamma^{\epsilon,L}$ satisfies
\[
\Ent_{\gamma^{\epsilon,L}}^{\Psi}(F)\le \frac{|A_{\epsilon}^{-1}|_{\textnormal{op}}}{2} \int_{\R^{\Lambda_{\epsilon,L}}}\Psi''(F)|\nabla F|^2d\gamma^{\epsilon,L}.
\]
The transport method thus yields (cf. \cite[proof of Theorem 5.3]{MS1}):

\begin{theorem}{\textnormal{($\Psi$-Sobolev inequalities)}}
\label{thm:Sobolev}
Let $\Psi:\mathcal I\to \R$ be a divergence and let $F:\R^{\Lambda_{\epsilon,L}}\to\mathcal I$ be any continuously differentiable function satisfying $\int_{\R^{\Lambda_{\epsilon,L}}} F^2d\nu^{\epsilon,L}\in\mathcal I$. Then,
\[
\Ent_{\nu^{\epsilon,L}}^{\Psi}(F)\le [c^{\epsilon,L}]^2\frac{|A_{\epsilon}^{-1}|_{\textnormal{op}}}{2} \int_{\R^{\Lambda_{\epsilon,L}}}\Psi''(F)|\nabla F|^2d\nu^{\epsilon,L}.
\]
\end{theorem}
Next we describe the $p$-Poincar\'e inequalities which are a different generalization of the Poincar\'e inequality. The measure  $\gamma^{\epsilon,L}$ satisfies (by a change of variables $x\mapsto A_{\epsilon}x$ in \cite[Theorem 2.6]{addona2021equivalence}),
\[
\int_{\R^{\Lambda_{\epsilon,L}}}F^p d\gamma^{\epsilon,L}\le \alpha_p^p|A_{\epsilon}^{-1}|_{\textnormal{op}}^{p/2} \int_{\R^{\Lambda_{\epsilon,L}}}|\nabla F|^p d\gamma^{\epsilon,L},
\]
where 
\[
\alpha_p:=
\begin{cases}
(p-1)^{1/2}\quad\text{for }p\in [2,\infty)\\
\frac{\pi}{2}\quad\text{for }p\in [1,2).
\end{cases}
\]
The transport method thus yields (cf.  \cite[Theorem 5.4]{MS1}):
\begin{theorem}{\textnormal{($p$-Poinacr\'e inequalities)}}
\label{thm:ppoincare}
Let $p\in [1,\infty)$ and let $F:\R^{\Lambda_{\epsilon,L}}\to\R$  be any continuously differentiable function satisfying $\int_{\R^{\Lambda_{\epsilon,L}}} Fd\nu^{\epsilon,L}=0$ and $F,\nabla F\in L^p(\gamma^{\epsilon,L})$. Then,
\[
\int_{\R^{\Lambda_{\epsilon,L}}}F^p d\nu^{\epsilon,L}\le [c^{\epsilon,L}]^p\,\alpha_p^p |A_{\epsilon}^{-1}|_{\textnormal{op}}^{p/2}\int_{\R^{\Lambda_{\epsilon,L}}}|\nabla F|^p d\nu^{\epsilon,L}.
\]
\end{theorem}

The next theorem also follows by the transport method but the argument is a bit different from the ones in the preceding results. We start by recalling that the Gaussian isoperimetric inequality plays a crucial role in high-dimensional probability, e.g., in the concentration of measure phenomenon, \cite{ledoux1996isoperimetry,ledoux2001concentration}. By changing variables $x\mapsto A_{\epsilon}x$ in  \cite[Theorem 3.1]{borell1975brunn}, the Gaussian isoperimetric inequality can be stated as 
\[
\gamma^{\epsilon,L}\left(K+rB\right)\ge \Phi\left(a+\frac{r}{\sqrt{ |A_{\epsilon}^{-1}|_{\textnormal{op}}}}\right),
\]
where $K\subset \R^{\Lambda_{\epsilon,L}}$ is any Borel set, $B\subset \R^{\Lambda_{\epsilon,L}}$ is the unit ball, $r\ge 0$, $\Phi$ is the cumulative distribution function of the standard one-dimensional Gaussian, and $a$ is such that $\gamma^{\epsilon,L}(K)=\Phi(a)$. The transport method thus yield (cf. \cite[Theorem 5.5]{MS1}):

\begin{theorem}{\textnormal{(Isoperimetric comparison)}}
\label{thm:iso}
\[
\nu^{\epsilon,L}\left(K+rB\right)\ge \Phi\left(a+\frac{r}{c^{\epsilon,L}\sqrt{ |A_{\epsilon}^{-1}|_{\textnormal{op}}}}\right).
\]
\end{theorem}

Finally, we turn to the issue of comparing the eigenvalues of the generators associated with $\gamma^{\epsilon,L}$ and $\nu^{\epsilon,L}$.  The constant in the Poincar\'e inequality for $\nu^{\epsilon,L}$ is the reciprocal of the first eigenvalue of the Langevin semigroup generator $\mathcal L(\nu^{\epsilon,L}):=\Delta+\left(\nabla\log \frac{d\nu^{\epsilon,L}}{d\varphi},\nabla\right)$, associated with $\nu^{\epsilon,L}$. Assuming the eigenvalues of $\mathcal L(\nu^{\epsilon,L})$ are discrete, let $\lambda_i(\nu^{\epsilon,L})$, $\lambda_i(\gamma^{\epsilon,L})$ be the eigenvalues of the generators $\mathcal L(\nu^{\epsilon,L}), \mathcal L(\gamma^{\epsilon,L})$, respectively. By Theorem \ref{thm:Sobolev}, $\lambda_1(\gamma^{\epsilon,L})\le c ^{\epsilon,L}\lambda_1(\nu^{\epsilon,L})$. A result of E. Milman \cite[Theorem 1.7]{milman2018spectral} shows that the transport method can be used to compare all higher-order eigenvalues. The transport method thus yields (cf. \cite[Corollary 3]{MS2}):

\begin{theorem}{\textnormal{(Eigenvalues comparisons)}}
\label{thm:eigenvalues}
For every $i\in \mathbb Z_+$,
\[
\lambda_i(\gamma^{\epsilon,L})\le [c^{\epsilon,L}]^2\, \lambda_i(\nu^{\epsilon,L}).
\]
\end{theorem}
Let us note, however, that in contrast to Theorems \ref{thm:Sobolev}--\ref{thm:iso} which are dimension-free in nature, Theorem \ref{thm:eigenvalues} depends on the dimension; even the discreteness and multiplicities of the eigenvalues are dimensional phenomena.

\subsection*{Acknowledgments} Many thanks to Roland Bauerschmidt for interesting conversations on the topic of this paper as well as for providing feedback on this work. I also thank Dan Mikulincer (for telling me about \cite{barashkov2021variational}), Max Raginsky for careful comments that improved this manuscript, and the anonymous referee for helpful suggestions. This material is based upon work supported by the National Science Foundation under Award Number 2002022. No data is associated with this manuscript.

\section{The Langevin transport map}
\label{sec:LN}
Throughout this section we assume that $\nu^{\epsilon,L}$ is smooth (an assumption which can often be removed, see \cite{MS2}) and satisfies the multiscale Bakry-\'Emery criterion of Definition \ref{def:BE}.
Let $\mu^{\epsilon,L}=\gamma^{\epsilon,L}$ and let $(U_t)_{t\ge 0}$ be the standard Ornstein-Uhlenbeck semigroup associated to the Ornstein-Uhlenbeck dynamics \eqref{eq:Langevin},
\[
d\Phi_t=\nabla\log \left(\frac{d\gamma^{\epsilon,L}}{d\varphi}\right)(\Phi_t)dt+\sqrt{2} dB_t,\quad \Phi_0\sim \nu^{\epsilon,L},
\]
satisfying, for any test function $F:\R^{\Lambda_{\epsilon,L}} \to \R$,
\begin{align}
\label{eq:eqOU}
\partial_tU_tF(\varphi)=\Delta U_tF(\varphi)-\left(A_{\epsilon}\varphi,\nabla U_tF(\varphi)\right).
\end{align}
The continuity (Fokker-Planck) equation for $p_t:=\text{Law}(\Phi_t)$ reads 
\[
\partial_tp_t=\nabla\left(p_t\nabla\log \frac{dp_t}{d\gamma^{\epsilon,L}}\right),\quad p_0=\nu^{\epsilon,L},
\]
and, using the equation \eqref{eq:eqOU}, it can be checked that $p_t=(U_te^{-V_0})\gamma^{\epsilon,L}$, where we recall that $e^{-V_0}=\frac{d\nu^{\epsilon,L}}{d\gamma^{\epsilon,L}}$. Letting 
\[
u_t(\varphi):=\log U_te^{-V_0(\varphi)},
\]
we see that the vector field $\nabla u_t$ drives the flow $(p_t)_{t\ge 0}$ via the continuity equation.

In order to show how the multiscale Bakry-\'Emery criterion can be used to control $\nabla u_t$, we will provide an explicit solution of \eqref{eq:eqOU}. Recall first that $Q_t:=e^{-t\frac{A_{\epsilon}}{2}}$, $\dot{C}_t:=Q_t^2=e^{-tA_{\epsilon}}$, and $C_t:=\int_0^t\dot{C}_sds$ so $C_{\infty}:=\int_0^{\infty}\dot{C}_sds=A_{\epsilon}^{-1}$.
\begin{lemma}
\label{lem:OU}
For any test function $F:\R^{\Lambda_{\epsilon,L}} \to \R$,
\begin{align*}
\label{eq:OU}
U_{\frac{t}{2}}F(\varphi)=\EE_{C_{t}}\left[F(e^{-t\frac{A_{\epsilon}}{2}}\varphi+\zeta)\right]=\EE_{C_{t}}\left[F(Q_{t}\varphi+\zeta)\right].
\end{align*}
\end{lemma}
\begin{proof}
By \cite[equation (4.1)]{faris2001ornstein} (where we take, using the notation is \cite{faris2001ornstein}, $Q\mapsto\Id$ and $A\mapsto \frac{A_{\epsilon}}{2}$), 
\[
\partial_t\EE_{C_t}\left[F(Q_t\varphi+\zeta)\right]=\frac{1}{2}\Delta \EE_{C_t}\left[F(Q_t\varphi+\zeta)\right]-\frac{1}{2}\left(A_{\epsilon}\varphi,\nabla \EE_{C_t}\left[F(Q_t\varphi+\zeta)\right]\right).
\]
On the other hand, by \eqref{eq:eqOU},
\[
\partial_t[U_{\frac{t}{2}}F(\varphi)]=\frac{1}{2}\Delta U_{\frac{t}{2}}F(\varphi)-\frac{1}{2}\left(A_{\epsilon}\varphi,\nabla U_{\frac{t}{2}}F(\varphi)\right),
\]
so the result follows as $U_0F(\varphi)=F(\varphi)=\EE_{C_0}\left[F(Q_0\varphi+\zeta)\right]$.
\end{proof}
We can now show how the multiscale Bakry-\'Emery criterion can be used to control $\nabla u_t$. By \eqref{eq:Vt} and Lemma \ref{lem:OU}, 
\[
V_t(Q_{t}\varphi)=-\log \EE_{C_{t}}\left[e^{-V_0}(Q_{t}\varphi+\zeta)\right]=-\log U_{\frac{t}{2}}e^{-V_0}(\varphi)=-u_{\frac{t}{2}}(\varphi),
\]
so, for any $\varphi\in \R^{\Lambda_{\epsilon,L}}$ and $t\ge 0$,
\[
\nabla u_{\frac{t}{2}}(\varphi)=-Q_t\nabla V_t(Q_t\varphi)\quad\text{and} \quad \nabla^2 u_{\frac{t}{2}}(\varphi)=-Q_t\nabla^2 V_t(Q_t\varphi)Q_t.
\]
In particular, if $\nu^{\epsilon,L}$ satisfies the multiscale Bakry-\'Emery criterion, then,  for any $\varphi\in \R^{\Lambda_{\epsilon,L}}$ and $t\ge 0$,
\begin{align}
\label{eq:BE}
-\nabla^2 u_t(\varphi)	\succeq \dot{\lambda}_{2t}\Id.
\end{align}

We now turn to the construction of the transport maps based on the vector field $\nabla u_t$. Define the family of maps $S_t:\R^{\Lambda_{\epsilon,L}}\to \R^{\Lambda_{\epsilon,L}}$ via the equation
\begin{align}
\label{eq:transport}
\partial_t S_t(\varphi)=-\nabla u_t\left(S_t(\varphi)\right),\quad S_0=\Id.
\end{align}
Arguing as in \cite[section 2]{MS2}, and using \eqref{eq:BE}, we get that $T_t:=S_t^{-1}$ is a diffeomorphism pushing forward $p_t$ to $p_0$. Moreover, $p_t\to \gamma^{\epsilon,L}$ weakly and $T^{\epsilon,L}:=\lim_{t\uparrow \infty}T_t$ pushes forward $\gamma^{\epsilon,L}=p_{\infty}$ to $\nu^{\epsilon,L}=p_0$. In addition, if $T_t$ is $L_t$-Lipschitz, then $T^{\epsilon,L}$ is $\limsup_{t\uparrow\infty}L_t$-Lipschitz, provided that the limit is finite.

\begin{theorem}
\label{thm:LN}
Set $\lambda_t:=\int_0^t\dot{\lambda}_sds$ for $t\in [0,\infty]$. The transport map $T_t$, which pushes forward $p_t$ to $\nu^{\epsilon,L}$, is $e^{\frac{-\lambda_{2t}}{2}}$-Lipschitz. In particular, the transport map $T^{\epsilon,L}$, which pushes forward $\gamma^{\epsilon,L}$ to $\nu^{\epsilon,L}$, is $e^{-\frac{\lambda_{\infty}}{2}}$-Lipschitz.
\end{theorem}

\begin{proof}
By \eqref{eq:transport}, for any $\varphi,w\in \R^{\Lambda_{\epsilon,L}}$, we have
\begin{align}
\label{eq:ODEderivative}
\partial_t[\nabla S_t(\varphi)w]=-\nabla^2 u_t(S_t(\varphi))[\nabla S_t(\varphi)w].
\end{align}
Our goal is to show that
\begin{align}
\label{eq:Lipschitz}
|T_t(\psi)-T_t(\varphi)|\le e^{\frac{-\lambda_{2t}}{2}}|\psi-\varphi|\quad\text{for all } \psi,\varphi\in \R^{\Lambda_{\epsilon,L}},
\end{align}
and our proof will follow the argument in \cite{MS2}. In order to establish \eqref{eq:Lipschitz}, it suffices to show that, for any unit $w\in \R^{\Lambda_{\epsilon,L}}$,
\begin{align}
\label{eq:LipschitzTemp}
|\nabla S_t(\varphi)w|\ge e^{\frac{\lambda_{2t}}{2}}.
\end{align}
 Indeed, if \eqref{eq:LipschitzTemp} holds, then $\nabla S_t(\varphi)(\nabla S_t(\varphi))^{\T}\ge \exp (\lambda_{2t})\Id$ so the inverse function theorem gives $\nabla T_t(\varphi)(\nabla T_t(\varphi))^{\T}\le \exp (-\lambda_{2t})\Id$. It follows that $|\nabla T_t(\varphi)|_{\text{op}}\le e^{-\frac{\lambda_{2t}}{2}}$, which is equivalent to \eqref{eq:Lipschitz}. 

In order to verify \eqref{eq:LipschitzTemp}, we will make use of \eqref{eq:ODEderivative} and \eqref{eq:BE}. Fix $w\in  \R^{\Lambda_{\epsilon,L}}$ and $\varphi\in \R^{\Lambda_{\epsilon,L}}$, with $|w|=1$, and define $\alpha_w(s):=\nabla S_s(\varphi)w$. Then, by \eqref{eq:ODEderivative} and \eqref{eq:BE},
\begin{align*}
\partial_s |\alpha_w(s)|&=\frac{1}{|\alpha_w(s)|}\alpha_w(s)^{\T}\partial_s\alpha_w(s)=\frac{1}{|\alpha_w(s)|}w^{\T}\nabla S_s(\varphi)^{\T}(-\nabla^2 u_s(S_s(\varphi)))\nabla S_s(\varphi)w\\
&\ge \dot{\lambda}_{2s}w^{\T}\nabla S_s(\varphi)^{\T}\nabla S_s(\varphi)w=\dot{\lambda}_{2s}|\nabla S_s(\varphi)w|=\dot{\lambda}_{2s}|\alpha_w(s)|.
\end{align*}
Since $|\alpha_w(0)|=1$, we can deduce from Gr\"onwall's inequality that
\[
\nabla S_t(\varphi)w=|\alpha_w(t)|\ge \exp\left(\int_0^t\dot{\lambda}_{2s}ds\right)=e^{\frac{\lambda_{2t}}{2}},
\]
which is \eqref{eq:LipschitzTemp}.
\end{proof}

\section{The Brownian transport map}
\label{sec:BM}
In this section we briefly sketch another connection between Polchinski's exact renormalization group and transportation of measures. Unlike the rest of the paper where we took $\tau=\infty$ and $\dot{C}_t=e^{-tA_{\epsilon}}$,  in this section we allow for arbitrary terminal time $\tau$ and positive semidefinite matrices $(\dot{C}_t)$. This freedom corresponds to the freedom in choice of cut-off functionals in the renormalization flow, or alternatively, to a change of metric. 

The way  the different versions of the multiscale Bakry-\'Emery criteria are used in \cite{bauerschmidt2020spectral, BB, bauerschmidt2022log, bauerschmidt2022logPhi} is the via the Polchinski semigroup, which we now present. Fix $\tau\in [0,\infty]$ and let $[0,\tau]\ni t \mapsto\dot{C}_t$ be a bounded function taking values in the set of positive semidefinite matrices such that, with $C_t:=\int_0^t\dot{C}_sds$, we have $C_{\tau}=A_{\epsilon}^{-1}$. Consider the stochastic differential equation \cite[equation (2.12)]{BB},
\begin{align}
\label{eq:FP}
d\tilde\Phi_t=-\dot{C}_{\tau-t}\nabla V_{\tau-t}(\tilde\Phi_t)dt+\dot{C}_{\tau-t}^{1/2}dB_t,\quad t\in [0,\tau],
\end{align}
where $(B_t)$ is a standard Brownian motion in $\R^{\Lambda_{\epsilon,L}}$ and $(V_t)$ is defined by \eqref{eq:Vt}. Denoting $\Phi_t:=\tilde\Phi_{\tau-t}$, we set \cite[equation (2.13)]{BB}, for $s\le t$,
\[
\PP_{s,t}F(\varphi):=\mathbb E[F(\Phi_s)|\Phi_t=\varphi], 
\] 
and this time-inhomogeneous semigroup $(\PP_{s,t})$ is the \textbf{Polchinski semigroup}. The process $(\Phi_t)$ is an instance of the Schr\"odinger bridge \cite{leonard2013survey, chen2021stochastic}, or the F\"ollmer process \cite{follmer1985entropy, follmer1986time,dai1991stochastic}, and  is obtained by taking the martingale $M_t:=\int_0^t\dot{C}_r^{1/2}dB_r$, which satisfies $M_{\tau}\sim \gamma^{\epsilon,L}$, and conditioning it so that $M_{\tau}\sim \nu^{\epsilon,L}$. The resulting process $(\Phi_t)$ of this conditioning, which interpolates $\Phi_0\sim\delta_0$ to $\Phi_{\tau}\sim\nu^{\epsilon,L}$, is an instance of Doob's $h$-transform \cite{jamison1975markov,cattiaux2014semi}. This process is also used in the \textbf{stochastic control} approach to Euclidean quantum field theories developed by Barashkov and Gubinelli \cite{barashkov2022stochastic, barashkov2020variational,barashkov2021variational}.

In the context of functional inequalities, the process $(\Phi_t)$ was developed independently by Eldan \cite{eldan2013thin}, from a different perspective called  \textbf{stochastic localization}, and by Lehec \cite{lehec2013representation} who used the special case $\tau=1$ and $\dot{C}_t=\Id$ . See also \cite[section 2.4.2]{chen2022localization} and \cite[Lemma 4.1]{MS1} (treating the case $\tau=1$ and $\dot{C}_t=\Id$). The work of Mikulincer and the author \cite{MS1} gave a new interpretation of the process $(\Phi_t)$ (in the case $\tau=1$ and $\dot{C}_t=\Id$) as the \textbf{Brownian transport map} which maps the Wiener measure on Wiener space into $\nu^{\epsilon,L}$. Further, by implicitly using one of the versions of multiscale Bakry-\'Emery criteria, the work \cite{MS1} deduces Lipschitz properties of the Brownian transport map. In essence, the measure $\mu^{\epsilon,L}$ from section \ref{subsec:transportfunct} is taken to be the Wiener measure (rather than $\gamma^{\epsilon,L}$) and the map $T^{\epsilon,L}$ is taken to be the Brownian transport map (rather than the Langevin/Ornstein-Uhlenbeck transport map). The idea behind this approach is that we can use infinite-dimensional Gaussian measures, e.g. the Wiener measure, since Gaussian measures satisfy functional inequalities with constants that are independent of the dimension; in our case this amounts to being independent of the infrared and ultraviolet cutoffs. We leave for future work the extension of the Brownian transport maps for general $(\dot{C}_t)$ and the incorporation of the various multiscale Bakry-\'Emery criteria.

\bibliographystyle{amsplain0}
\bibliography{refPolchinski}

\end{document}